\documentclass{article}

\usepackage{hyperref}
\usepackage{tikz}
\usepackage{amsthm}
\usepackage{amsmath}
\usepackage{amsfonts}
\usepackage{algorithm}
\usepackage{enumerate}

\newtheorem{defn}{Definition}
\newtheorem{prob}{Problem}
\newtheorem{thm}{Theorem}
\newtheorem{lemma}{Lemma}
\newtheorem{obs}{Observation}

\newcommand{\R}{\mathbb R}
\newcommand{\C}{\mathbb C}
\newcommand{\1}{\mathbf 1}
\newcommand{\CC}{\mathcal C}

\DeclareMathOperator{\argmax}{\mathrm{arg\,max}}
\DeclareMathOperator{\expval}{\mathrm E}
\DeclareMathOperator{\prb}{\mathrm{Pr}}
\DeclareMathOperator{\diag}{\mathrm{diag}}
\DeclareMathOperator{\rk}{\mathrm{rk}}
\DeclareMathOperator{\matrdim}{\mathrm{dim}}
\DeclareMathOperator{\var}{\mathrm{Var}}
\DeclareMathOperator{\PP}{\mathcal{PP}}

\title{Recovering Nonuniform Planted Partitions via Iterated Projection}
\author{Sam Cole}

\begin{document}

\maketitle

\begin{abstract}
In the \emph{planted partition problem}, the $n$ vertices of a random graph are partitioned into $k$ ``clusters,'' and edges between vertices in the same cluster and different clusters are included with constant probability $p$ and $q$, respectively (where $0 \le q < p \le 1$).  We give an efficient spectral algorithm that recovers the clusters with high probability, provided that the sizes of any two clusters are either very close or separated by $\geq \Omega(\sqrt n)$.  We also discuss a generalization of planted partition in which the algorithm's input is not a random graph, but a random real symmetric matrix with independent above-diagonal entries.

Our algorithm is an adaptation of a previous algorithm for the uniform case, i.e., when all clusters are size $n / k \geq \Omega(\sqrt n)$.  The original algorithm recovers the clusters one by one via \emph{iterated projection}: it constructs the orthogonal projection operator onto the \emph{dominant $k$-dimensional eigenspace} of the random graph's adjacency matrix, uses it to recover one of the clusters, then deletes it and recurses on the remaining vertices.  We show herein that a similar algorithm works in the nonuniform case.
\end{abstract}

\section{Introduction}

In the \emph{planted partition problem}, $n$ fixed vertices are partitioned into $k$ unknown ``clusters'' $C_1, \ldots, C_k$, and edges are added independently with probability $p$ between pairs of vertices in the same cluster and probability $q$ between vertices in different clusters, where $p$ and $q$ are constants such that $0 \leq q < p \leq 1$.  The goal is then to recover the unknown partition a.s.\ given a random graph drawn from this distribution.

This paper is a companion to~\cite{ColeFR17}, which gives a simple spectral algorithm for the special case of planted partition in which $|C_i| = n / k$ for all $i$ and $n / k = \Omega(\sqrt n)$.  
Our algorithm recovers the unknown partition via \emph{iterated projection}: it constructs the orthogonal projection operator onto the dominant $k$-dimensional eigenspace of the adjacency matrix of the randomly generated graph and uses it to recover a single cluster, then deletes it and recurses.

In this paper, we show that, with minor modifications, the same algorithm works in a much more general setting: namely, the setting in which the clusters are partitioned into ``superclusters,'' where clusters in the same supercluster are approximately the same size, while clusters in different superclusters have sizes separated by $ \geq \Omega(\sqrt n)$ (and, as in the uniform case, all clusters are size $\geq \Omega(\sqrt n)$).  

\subsection{Outline}

In Section~\ref{plantedpartition} we formally define the planted partition problem.  In Section~\ref{uniformcase} we briefly review the uniform case.  In Section~\ref{superclusters} we define the superclusters setting discussed above.  In Section~\ref{sec:algorithm} we describe our algorithm, and in Sections~\ref{BBhateigs}-\ref{delandrec} we prove its correctness.  In Section~\ref{paramfree} we show how to estimate the number of clusters in each supercluster empirically if the exact numbers are not known.  Finally, in Section~\ref{gen} we discuss a generalization of planted partition in which the algorithm's input is not a random graph, but a random real symmetric matrix $\hat A = (\hat a_{uv})_{u, v = 1}^n$ such that $\hat a_{uv}$ are independent random variables for $1 \leq i \leq j \leq n$ with expectation $p$ or $q$ (depending on whether $u$ and $v$ are in the same cluster or not).

\section{The planted partition problem}\label{plantedpartition}

We now formally define the planted partition problem.

\begin{defn}[Planted partition model]
Let $\mathcal C = \{C_1, \ldots, C_k\}$ be a partition of the set $[n] := \{1, \ldots, n\}$ into $k$ sets called \emph{clusters}, with $|C_i| =: s_i$ for $i = 1, \ldots, k$.  For constants 
$0 \leq q < p \leq 1$, we define the \emph{planted partition model} $\mathcal G(n, \mathcal C, p, q)$ to be the probability space of graphs with vertex set $[n]$, with edges $uv$ (for $u \neq v$) included independently with probability $p$ if $u$ and $v$ are in the same cluster in $\mathcal C$ and probability $q$ otherwise.  
\end{defn}

See Figure~\ref{plantedpartitionfigure}.  Note that the case $k = 1$ gives the standard Erd\H{o}s-R\'enyi model $\mathcal G(n, p)$~\cite{er1959}, and the case $k = n$ gives $\mathcal G(n, q)$.
\begin{figure}
\begin{center}
\begin{tikzpicture}[scale=.08]

\foreach \x/\y/\z/\w in {0/34.64/20/0, 20/0/-20/0, -20/0/0/34.64} {	

\draw (\x, \y) circle(16);	

\fill (\x - 7, \y - 2) circle(1);
\draw (\x - 7, \y - 2) -- (\x + 5, \y - 5);
\fill (\x + 5, \y - 5) circle(1);
\draw (\x - 1, \y - 3.5) node[anchor=south]{$p$};

\foreach \t in {.35, .65}
\fill (\t * \x + \z - \t * \z, \t * \y + \w - \t * \w) circle(1);
\draw[dotted] (.35 * \x + .65 * \z, .35 * \y + .65 * \w) -- (.65 * \x + .35 * \z, .65 * \y + .35 * \w);
\draw (.5 * \x + .5 * \z, .5 * \y + .5 * \w) node[anchor=north]{$q$};
}

\draw (0, 51.64) node[anchor=north]{$C_1$};
\draw (37, 0) node[anchor=east]{$C_2$};
\draw (-37, 0) node[anchor=west]{$C_3$};

\end{tikzpicture}
\end{center}
\caption{An illustration of the planted partition model.  Edges between two vertices in the same cluster are added with probability $p$, while edges between two vertices in different clusters are added with probability $q$.}\label{plantedpartitionfigure}

\end{figure}
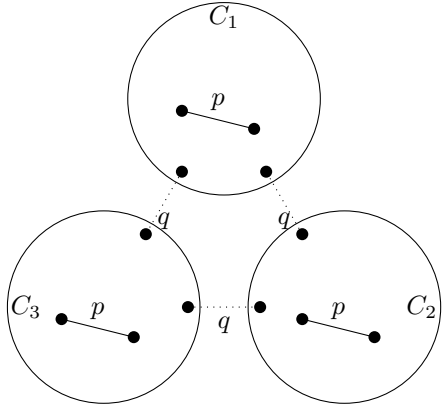

\begin{prob}[Planted partition]
Identify (or ``recover'') the unknown partition $C_1, \ldots, C_k$ (up to a permutation of $[k]$) given only a random graph $\hat G \sim \mathcal G(n, \mathcal C, p, q)$.
\end{prob}

Observe that, by considering the adjacency matrix of $\hat G$, we can think of this as a problem about random symmetric matrices whose above-diagonal entries are independent Bernoulli random variables.

\section{The uniform case}\label{uniformcase}

We now present a slightly modified version of the algorithm presented in~\cite{ColeFR17} for the uniform case:

\begin{algorithm}[H]
\caption{Uniform iterated projection}\label{alg}
Given a graph $\hat G = (\hat V, \hat E)$, cluster size $s$:
\begin{enumerate}
\item	Let $\hat A$ be the adjacency matrix of $\hat G$, $n := |\hat V|$, $k := n / s$. 

\item	Let $P_k(\hat A) =: (\hat p_{uv})_{u, v \in \hat V}$ be the orthogonal projection operator onto the subspace of $\R^n$ spanned by eigenvectors corresponding to the largest $k$ eigenvalues of $\hat A$.\label{step:proj}

\item	For each column $v$ of $P_k(\hat A)$, let $W_v := \{u \in \hat V : \hat p_{uv} \geq \frac1{2s}\}$, i.e., the indices of the ``large'' entries of column $v$ of $P_k(\hat A)$.\label{step:Wj}

\item	Let $v^*$ be the column $v$ such that $|W_v| \leq (1 + \epsilon)s$ and $||P_k(\hat A)\1_{W_v}||_2$ is maximum, i.e. $v^* := \argmax_{v : |W_v| \leq (1 + \epsilon)s}||P_k(\hat A)\1_{W_v}||_2$.  It will be shown that such a $v^*$ exists and $W_{v^*}$ has large intersection with a single cluster $C_i \in \mathcal C$ a.s.\label{step:jstar}

\item	Let $C$ be the set of vertices in $\hat G$ with $\geq (p - 10\epsilon)s$ neighbors in $W_{v^*}$.  It will be shown that $C = C_i$ a.s.\label{step:id}

\item	Remove $C$ and repeat on $\hat G[\hat V \setminus C]$.  Stop when there are $< s$ vertices left.
\end{enumerate}

\end{algorithm}

The overview of Algorithm~\ref{alg} is as follows.  The algorithm gets a random graph $\hat{G}$ generated according to $\mathcal G(n, \mathcal C, p ,q)$.  We first construct the projection operator which projects onto the subspace of $\R^n$ spanned by the eigenvectors corresponding to the largest $k$ eigenvalues of $\hat{G}$'s adjacency matrix.  This, we will argue, gives a fairly good approximation of at least one of the clusters, which we can then find and ``fix up."  Then we remove the cluster and repeat the algorithm.

The main result in~\cite{ColeFR17} is that this algorithm a.s.\ recovers planted partitions in which all clusters are the same size $s = \Omega(\sqrt n)$:

\begin{thm}\label{mainthm}
For sufficiently large $n$ with probability $\geq 1 - 2^{-\Omega(\sqrt n)}$, Algorithm~\ref{alg} correctly recovers planted partitions in which all clusters are size $s \geq c\sqrt n$, where $c = c(p, q) = \Theta((p - q)^{-2})$. 
\end{thm}

\section{A more general setting}\label{superclusters}

Without much work, one can show that, in fact, Algorithm~\ref{alg} works when all clusters are \emph{almost} the same size---i.e., when $(1 - \epsilon)\frac nk \leq |C_i| \leq (1 + \epsilon)\frac nk$ for all $i$, where $\epsilon = O(p - q)$.  A natural next step is to try to extend it to the case when the clusters are divided into $K$ ``superclusters,'' where clusters in the same supercluster have roughly the same size, while clusters in different superclusters have sizes separated by $\geq c\sqrt n$.  This is the setting which we consider for the remainder of this paper.

More precisely:
\begin{itemize}
\item	Let $\CC = \{C_1, \ldots, C_k\}$ be the set of clusters.  
\item	Let $s_i := |C_i|$ for $i = 1, \ldots, m$ and assume without loss of generality that
\begin{equation}\label{siassumption}
s_1 \geq \ldots \geq s_m \geq c\sqrt n.
\end{equation}  
\item	Assume $\CC$ is partitioned into $K$ ``superclusters'' $\CC = \CC_1 \cup \ldots \cup \CC_K$.
\item	Let $k_i := |\CC_i|$ be the number of clusters in supercluster $\CC_i$, for $i = 1, \ldots, K$.  \item	Assume that the sizes of clusters in different superclusters are separated by $\geq c\sqrt n$.  Furthermore, we may assume that the $\CC_i$ are arranged in decreasing order of their cluster sizes; i.e.,
\begin{equation}\label{superclustersep}
\min_{C \in \CC_i} |C| \geq \max_{C \in \CC_{i + 1}}|C| + c\sqrt n
\end{equation}
for $i = 1, \ldots, K - 1$.  Thus, by~\eqref{siassumption} we have
\begin{equation}\label{superclustersepindices}
\CC_1 = \{C_1, \ldots, C_{k_1}\}, \CC_2 = \{C_{k_1 + 1}, \ldots, C_{k_1 + k_2}\}, \ldots, \CC_K = \{C_{k - k_K + 1}, \ldots, C_k\}.
\end{equation}
\item	Within the superclusters the sizes are approximately the same:
\begin{equation}\label{withinsupercluster}
\max_{C \in \CC_i}|C| \leq (1 + \epsilon)\min_{C \in \CC_i}|C|
\end{equation}
for $i = 1, \ldots, K$, where $\epsilon = \epsilon(p, q)$ will be specified later.
\item	We may sometimes abuse notation and use $\mathcal C_i$ to refer to the set of \emph{indices} $j$ such that $C_j \in \mathcal C_i$ or the set of \emph{vertices} $u \in \bigcup_{C \in \mathcal C_i}C$.
\end{itemize}  

Our goal is still to recover the \emph{individual clusters} exactly; we do not care about identifying which pairs of clusters belong to the same supercluster.  Note that we assume that we know $p$, $q$, and $k_1, \ldots, k_K$ a priori.  We will discuss how to determine these parameters empirically in Section~\ref{paramfree}.

\subsection{Notation and definitions}\label{notation}

We will use the following graph and matrix notation throughout this paper:

\begin{itemize}
\item	$N_G(v)$ -- neighborhood of vertex $v$ in a graph $G$.  We will omit the subscript $G$ when the meaning is clear.
\item	$G[S]$ -- the induced subgraph of $G$ on $S \subseteq V(G)$.
\item	$A[S]$ -- the principal submatrix of $A$ with row and column indices restricted to $S$.
\item	$\lambda_i(A)$ -- the $i$th largest eigenvalue of a symmetric matrix $A$ (recall that symmetric matrices have real eigenvalues).
\item	$\lambda_i(G)$ -- the $i$th largest eigenvalue of $G$'s adjacency matrix.
\item	$||\cdot||_2$ -- the $\ell_2$- (``spectral'') norm of a vector or matrix.
\item	$||\cdot||_F$ -- the Frobenius norm of a matrix.
\item	$I_n$ -- the $n \times n$ identity matrix.
\item	$J_n$ -- the $n \times n$ 1s matrix.
\item	$\1_S$ -- the indicator vector $\in \{0, 1\}^n$ for the set $S \subseteq [n]$.
\item	$\1_n$ -- the all 1s vector $\in \R^n$, i.e. $\1_{[n]}$.
\item	$\expval[X]$ -- the expectation of a random variable $X$.  If $X$ is matrix or vector valued, then the expectation is taken entrywise.
\item	a.s.\ -- almost surely, i.e.\ with probability $1 - o(1)$ as $n \to \infty$.
\end{itemize}

In addition, the following definitions are central to our algorithm:

\begin{defn}[Dominant eigenspace]
The \emph{dominant $r$-dimensional eigenspace} of an $n \times n$ symmetric matrix $A$ is the subspace of $\R^n$ spanned by eigenvectors corresponding to the largest $r$ eigenvalues of $A$.  Note that this is well-defined as long as $\lambda_r(A) \neq \lambda_{r + 1}(A)$.
\end{defn}

\begin{defn}[Rank-$r$ projector]\label{rankrproj}
The \emph{rank-$r$ projector} of an $n \times n$ symmetric matrix $A$, denoted $P_r(A)$, is the orthogonal projection operator onto the dominant $r$-dimensional eigenspace of $A$, represented in the standard basis for $\R^n$.
\end{defn}

We will denote as follows the main quantities to consider in this paper. 
\begin{itemize}
\item	$\hat{G} = ([n], \hat E)$ -- a random graph obtained from an \emph{unknown} planted partition distribution $\mathcal G(n, \mathcal C, p, q)$.  This is what the cluster identification
algorithm receives as input. 
\item	$\hat A = (\hat a_{uv})_{u, v = 1}^n \in \{0, 1\}^{n \times n}$ -- the adjacency matrix of $\hat G$.
\item	$\expval[\hat A] := (\expval[\hat a_{uv}])_{u, v = 1}^n$ -- the entrywise expectation of $\hat A$.  \item	$A = (a_{uv})_{u, v= 1}^n := \expval[\hat A] + pI_n$ -- the expectation of the adjacency matrix $\hat{G}$ with $p$s added to the diagonal (to
make it a rank $k$).
\item	$\hat B = (\hat b_{uv})_{u, v= 1}^n:= \hat A + pI_n - qJ_n$.
\item	$B = (b_{uv})_{u, v= 1}^n:= \expval[\hat B] = A - qJ_n$.
 \end{itemize}
 
 A key difference between the uniform and nonuniform cases is that the algorithm and analysis are based on $B$ and $\hat B$ rather than $A$ and $\hat A$.  This simplifies the spectral analysis considerably, since $B$ is essentially a block diagonal matrix (after permuting the rows and columns).  In order to compute $\hat B$, we assume that our algorithm has access to the exact values of $p$ and $q$, or at least good approximations.  We discuss this further in Section~\ref{unknownpq}.
 
\section{The algorithm}\label{sec:algorithm}

We now show how to adapt Algorithm~\ref{alg} to the ``superclusters'' setting presented in Section~\ref{superclusters}.  The key difference is that we will project onto the eigenspace of $\hat B$ corresponding to its largest $k_1$ (rather than $k$) eigenvalues.  Because we have an $\Omega(\sqrt n)$ separation between $\CC_1$ and $\CC_2$, this will allow us to recover one of the clusters in $\mathcal C_1$.  When we have recovered all clusters in $\CC_1$, we will move on to $\CC_2$, then $\CC_3$, and so on.

Another complication is that, since the clusters are not all the same size, it is no longer reasonable to assume we know the cluster sizes exactly.  However, we will see that the eigenvalues of $\hat B$ give good approximations to the cluster sizes.  More precisely, for $C_i \in \CC_j$, $\lambda_i(\hat B)$ is a good approximation to $s_i$.  In fact, since all clusters in $\CC_j$ are approximately the same size, it is a good approximation to the size of \emph{any} cluster in $\CC_j$.  This allows us to construct an approximate cluster of roughly the correct size, as in Step~\ref{step:jstar} of Algorithm~\ref{alg}.

\begin{algorithm}[H]
\caption{Nonuniform iterated projection}\label{nonunifalg}
Given a graph $\hat G = (\hat V, \hat E)$, supercluster sizes $k_1, \ldots, k_K$:
\begin{enumerate}
\item	Let $\hat A$ be the adjacency matrix of $\hat G$, $n := |\hat V|$, $\hat B := \hat A - qJ_n + pI_n$. 

\item	Let $P_{k_1}(\hat B) =: (\hat p_{uv})_{u, v \in \hat V}$ be the orthogonal projection operator onto the dominant $k_1$-dimensional eigenspace of $\hat B$.\label{step:proj:nonunif}

\item	Let $\hat s := (\lambda_1(\hat B) + 7\sqrt n) / (p - q)$.  We will see that this is approximately the size of the largest cluster.\label{step:shat}

\item	For each column $v$ of $P_{k_1}(\hat B)$, let $W_v := \{u \in \hat V : \hat p_{uv} \geq \frac1{2\hat s}\}$, i.e., the indices of the ``large'' entries of column $v$ of $P_{k_1}(\hat B)$.\label{step:Wj:nonunif}

\item	Let $v^*$ be the column $v$ such that $|W_v| \leq (1 + \epsilon)s$ and $||P_{k_1}(\hat B)\1_{W_v}||_2$ is maximum, i.e. $v^* := \argmax_{v : |W_v| \leq (1 + \epsilon)\hat s}||P_{k_1}(\hat B)\1_{W_v}||_2$.  It will be shown that such a $v^*$ exists and $W_{v^*}$ has large intersection with a single cluster $C_i \in \mathcal C_1$ a.s.\label{step:jstar:nonunif}

\item	Let $C$ be the set of vertices in $\hat G$ with $\geq (p - 10\epsilon)\hat s$ neighbors in $W_{v^*}$.  It will be shown that $C = C_i$ a.s.\label{step:id:nonunif}

\item	Remove $C$ and repeat on $\hat G[\hat V \setminus C]$, with supercluster sizes $k_1 - 1, k_2, \ldots, k_K$.  If $k_1 = 1$, instead use $k_2, \ldots, k_K$ as the supercluster sizes (i.e., $k_2$ becomes the ``new'' $k_1$, $k_3$ the ``new'' $k_2$, and so on).  Stop when all supercluster sizes are 0.
\end{enumerate}
\end{algorithm}

The main result of this paper is the following:

\begin{thm}
Let $\mathcal C$ be an unknown partition of $[n]$ satisfying the conditions in Section~\ref{superclusters}, with $\epsilon = O(p - q)$  and $c = \Omega\left(\frac1{(p - q)\epsilon}\right)$.  Then Algorithm~\ref{nonunifalg} recovers $\mathcal C$ given only $\hat G \sim \mathcal G(n, \mathcal C, p, q)$ with probability $\geq 1 - 2^{-\Omega(\sqrt n)}$.
\end{thm}

Sections~\ref{BBhateigs}-\ref{delandrec} are devoted to proving the correctness of Algorithm~\ref{nonunifalg}, mirroring the analysis in~\cite{ColeFR17}.  Sections~\ref{BBhateigs} and~\ref{projdev} develop the linear algebra tools necessary for the proof, Section~\ref{recsingle} uses these tools to prove that Steps~\ref{step:Wj:nonunif}-\ref{step:id:nonunif} of Algorithm~\ref{nonunifalg} successfully recover a single cluster a.s., while Section~\ref{delandrec} shows that the algorithm as a whole successfully recovers \emph{all} clusters a.s.  

\section{Eigenvalues of $B$ and $\hat B$}\label{BBhateigs}

Observe that by permuting the rows and columns of $B$ we get $B \sim (p - q)\diag(J_{s_1}, \ldots, J_{s_k})$.  Thus, its eigenvalues are trivial to compute:

\begin{lemma}\label{Beigs}
B is a rank-$k$ matrix with eigenvalues 
\begin{eqnarray*}
&&\lambda_i(B) = (p-q)s_i\ \textrm{for}\ i = 1,\ldots,k,	\\
&&\lambda_i(B) = 0\ \textrm{for}\ i = k + 1, \ldots, n.\end{eqnarray*}
\end{lemma}

As in the uniform case, we a.s.\ get a deviation of at most $O(\sqrt n)$ between the eigenvalues of $B$ and $\hat B$ by applying a modified version of F\"uredi and Koml\'os's well-known result on the distribution of eigenvalues of random symmetric matrices~\cite{FK81}:

\begin{lemma}\label{lemma:B-Bhatnorm}
With probability $\geq 1 - e^{-n}$,
\begin{equation}\label{B-Bhatnorm}
||B - \hat B||_2 \leq 7\sqrt n
\end{equation}
for sufficiently large $n$.
\end{lemma}

\begin{proof}
Apply~\cite[Theorem~7]{ColeFR17} to $X := \hat B - B$ with $K = 1, \sigma = 1 / 2$ to get
\[||B - \hat B||_2 = \max_{i = 1}^n|\lambda_i(X)| \leq 7\sqrt n\]
with probability $\geq 1 - e^{-n}$.
\end{proof}

By Weyl's inequalities (see, e.g.,~\cite[Theorem4.4.6]{friedland2015matrices}), we get
\begin{equation}\label{weyl}
|\lambda_i(B) - \lambda_i(\hat B)| \leq 7\sqrt n
\end{equation}
for $i = 1, \ldots, n$; i.e., we can approximate the eigenvalues of $B$ with those of $\hat B$ (and vice versa) with at most $O(\sqrt n)$ error.  This yields an $\Omega(\sqrt n)$ separation in the eigenvalues of both $B$ and $\hat B$ between different superclusters; i.e., for $i = 1, \ldots, K - 1$,
\[\min_{C_j \in \CC_i}\min\{\lambda_j(B), \lambda_j(\hat B)\} \geq \max_{C_j \in \CC_{i + 1}}\max\{\lambda_j(B), \lambda_j(\hat B)\} + \Omega(\sqrt n),\]
as shown in Figure~\ref{eigdist}.  However, such a separation between $\CC_1$ and $\CC_2$ will suffice, since Algorithm~\ref{nonunifalg} computes $P_{k_1}(\cdot)$ in \emph{each} iteration on a submatrix of the original $\hat B$ from the first iteration.

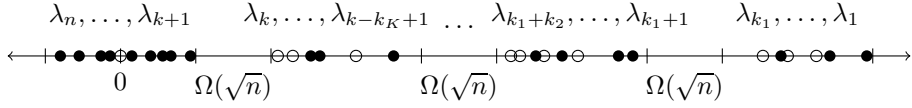
\begin{figure}
\centering

\begin{tikzpicture}

\draw[<->] (-1.5, 0) -- (10.5, 0);

\foreach \x in {-1, 0, 1, 2, 4, 5, 7, 8, 10}
\draw (\x, .1) -- (\x, -.1);

\foreach \x in {-.26, -.55, .4, -.8, .93, -.14, .15, .56, .67,
3.63, 2.65, 2.53, 
5.52, 5.87, 6.62, 6.81,
9.43, 8.78, 9.92}
\fill (\x, 0) circle(.075);

\foreach \x in {0, 3.13, 2.29, 2.09,
5.19, 5.63, 5.31, 6.08,
8.87, 9.25, 8.54}
\draw (\x, 0) circle(.075);

\foreach \x/\lbl in {0/{$\lambda_n, \ldots, \lambda_{k + 1}$}, 2.875/{$\lambda_k, \ldots, \lambda_{k - k_K + 1}$}, 6.25/{$\lambda_{k_1 + k_2}, \ldots, \lambda_{k_1 + 1}$}, 9/{$\lambda_{k_1}, \ldots, \lambda_1$}, 4.5/$\ldots$}
\draw (\x, .25) node[anchor = south]{\lbl};

\foreach \x in {1.5, 4.5, 7.5}
\draw (\x, -.1) node[anchor = north]{$\Omega(\sqrt n)$};

\draw (0, -.1) node[anchor = north]{0};

\end{tikzpicture}

\caption{The distribution of eigenvalues of $B$ ($\circ$) and $\hat B$ ($\bullet$)}\label{eigdist}
\end{figure}

\begin{lemma}\label{eigsep}
Assume~\eqref{B-Bhatnorm} holds.  Then 
\[(p - q)s_{k_1} - 7\sqrt n \leq \lambda_i(B), \lambda_i(\hat B) \leq (p - q)s_1 + 7\sqrt n\]
for $i = 1, \ldots, k_1$, while
\[\lambda_i(B), \lambda_i(\hat B) \leq (p - q)s_{k_1 + 1} + 7\sqrt n\]
for $i = k_1 + 1, \ldots, n$.
\end{lemma}

\begin{proof}
First we handle $B$.  For $i \leq k_1$ we have
\[(p - q)s_{k_1} \leq \lambda_i(B) = (p - q)s_i \leq (p - q)s_1.\]
By definition of $k_1$, we have $\CC_1 = \{C_1, \ldots, C_{k_1}\}$; thus, by~\eqref{superclustersep} we have
\[\lambda_i(B) = (p - q)s_i \leq (p - q)s_{k_1 + 1}\]
for $i = k_1 + 1, \ldots, k$, and for $i > k$ we have
\[\lambda_i(B) = 0 \leq (p - q)s_{k_1 + 1}.\]
The lemma thus follows by~\eqref{weyl}
\end{proof}

Thus, the eigenvalues of $B$ and $\hat B$ corresponding to the clusters in $\CC_1$ are separated from the remaining eigenvalues by at least $((p - q)c - 14)\sqrt n$, since $s_{k_1} \geq s_{k_1 + 1} + c\sqrt n$ by~\eqref{superclustersep}.  This quantity is positive as long as
\[c > \frac{14}{p - q}.\]

\section{Deviation between the projection operators}\label{projdev}

The following lemma shows an inverse relation between a ``gap'' in the eigenvalues of two symmetric matrices $X$ and $Y$ and the difference between their rank-$r$ projectors (Definition~\ref{rankrproj}).  More precisely, if the largest $r$ eigenvalues of both $X$ and $Y$ are ``well-separated'' from the remaining ones (as is the case for $B$ and $\hat B$ above, with $r = k_1$), then their projectors $P_r(X)$ and $P_r(Y)$ are close in $\ell_2$-norm.

\begin{lemma}\label{projdiff}
Let $X, Y \in \R^{n \times n}$ be symmetric.  Suppose that the largest $r$ eigenvalues of both $X$ and $Y$ are $\geq \beta$, and the remaining $n - r$ eigenvalues of both $X$ and $Y$ are $\leq \alpha$, where $\alpha < \beta$.  Then
\begin{equation}\label{projdiffl2}
||P_r(X) - P_r(Y)||_2 \leq \frac{||X - Y||_2}{\beta - \alpha}
\end{equation}
and
\begin{equation}\label{projdiffFrob}
||P_r(X) - P_r(Y)||_F \leq \frac{\sqrt{2r}||X - Y||_2}{\beta - \alpha}.
\end{equation}
\end{lemma}

This can be proved via the Cauchy integral formula for projections, as in~\cite[Lemmas~15-16]{ColeFR17}.  See Appendix~\ref{projdiffproof} for the full proof.  

We can apply the above lemma to $B$ and $\hat B$ to get the following:

\begin{lemma}\label{lemma:projl2norm}
Assume~\eqref{B-Bhatnorm} holds.  Then we have
\begin{equation}\label{projl2norm}
||P_{k_1}(B) - P_{k_1}(\hat B)||_2 \leq \epsilon
\end{equation}
and
\begin{equation}\label{projfrobnorm}
||P_{k_1}(B) - P_{k_1}(\hat B)||_F \leq \sqrt{2k_1}\epsilon,
\end{equation}
provided that
\begin{equation}\label{cepsilonrel}
c \geq \frac{14}{(p - q)\epsilon}.
\end{equation}
\end{lemma}

\begin{proof}
By~\eqref{B-Bhatnorm} and Lemma~\ref{eigsep}, we can apply Lemma~\ref{projdiff} with 
\[X = B,\quad Y = \hat B,\quad \alpha = (p - q)s_{k_1} - 7\sqrt n,\quad \beta = (p - q)s_{k_1 + 1} + 7\sqrt n\]
to get
\begin{eqnarray*}
||P_{k_1}(B) - P_{k_1}(\hat B)||_2 	& \leq	& \frac{7\sqrt n}{(p - q)(s_{k_1} - s_{k_1 + 1}) - 14\sqrt n} \\
	& \leq	& \frac{7\sqrt n}{(p - q)c\sqrt n - 14\sqrt n}	\\
	& \leq	& \frac{14}{(p - q)c}	\\
	& \leq	& \epsilon,
\end{eqnarray*}
where the second inequality follows from~\eqref{superclustersep}.  We get~\eqref{projfrobnorm} similarly by applying~\eqref{projdiffFrob}.
\end{proof}

Thus, we see that $c$ has an inverse dependence on $\epsilon$ and $p - q$.  In order for the proofs in Section~\ref{exactcluster} to go through, we will require that $\epsilon = O(p - q)$.  Thus, by~\eqref{cepsilonrel} $c$ must be $\Omega((p - q)^{-2})$.

Why is Lemma~\ref{lemma:projl2norm} useful?  Observe that
\[P_{k_1}(B) = \sum_{i = 1}^{k_1}\frac1{s_i}\1_{C_i}\1_{C_i}^\top.\]
Thus, the nonzero entries of $P_{k_1}(B)$ tell us precisely which pairs of vertices belong to the same cluster $C \in \CC_1$.  Of course our algorithm does not have access to this matrix, but by Lemma~\ref{lemma:projl2norm} $P_{k_1}(\hat B) \approx P_{k_1}(B)$ a.s., so we should be able to use $P_{k_1}(\hat B)$ in place of $P_{k_1}(B)$ to recover the clusters.  Sections~\ref{recsingle} and~\ref{delandrec} go over the details of this approach.

\section{Recovering a single cluster}\label{recsingle}

In this section we show how to use the spectral results in Sections~\ref{BBhateigs} and~\ref{projdev} to recover a single cluster.  In Section~\ref{approxcluster} we will show how to construct a.s.\ a set $W$ with large intersection with a single cluster $C_i$ (Steps~\ref{step:Wj:nonunif}-\ref{step:jstar:nonunif} of Algorithm~\ref{nonunifalg}), and in Section~\ref{exactcluster} we will show how to recover $C_i$ exactly a.s.\ by looking at the number of neighbors in $W$ of each vertex (Step~\ref{step:id:nonunif}).  In Section~\ref{delandrec} we will show how to recover \emph{all} clusters using this procedure.

\subsection{Constructing an approximate cluster}\label{approxcluster}

We do not assume that our algorithm has access to the exact cluster sizes, so let us begin by showing that $\hat s$ as defined in Step~\ref{step:shat} of Algorithm~\ref{nonunifalg} is  a good approximation to the size of the clusters in $\CC_1$ (recall that by~\eqref{superclusters} they are all approximately the same size).

\begin{lemma}
Assume~\eqref{B-Bhatnorm} holds and define  
\begin{equation*}\label{shatdef}
\hat s := \frac{\lambda_1(\hat B) + 7\sqrt n}{p - q}.
\end{equation*}
Then
\begin{equation}\label{shatbounds}
s_{k_1} \leq \ldots \leq s_1 \leq \hat s \leq (1 + 2\epsilon)s_{k_1}.
\end{equation}
\end{lemma}

\begin{proof}
By Lemma~\ref{Beigs}, $\lambda_1(B) = (p - q)s_1$, so by~\eqref{B-Bhatnorm} and Weyl's inequalities we have 
\[(p - q)s_1 - 7\sqrt n \leq \lambda_1(\hat B) \leq (p - q)s_1 + 7\sqrt n.\]
Thus, $\hat s$ is an upper bound on $s_1$.  Finally, as $s_1 \leq (1 + \epsilon)s_{k_1}$ by equation~\eqref{withinsupercluster}, we have
\begin{eqnarray*}
\hat s	& \leq	& \frac{\lambda_1(B) + 14\sqrt n}{p - q}	\\
		& =		& s_1 + \frac{14\sqrt n}{p - q}				\\
		& =		& s_{k_1}\left(\frac{s_1}{s_{k_1}} + \frac{14\sqrt n}{(p - q)s_{k_1}}\right)	\\
		& \leq	& s_{k_1}\left(1 + \epsilon + \frac{14}{(p - q)c}\right)				\\
		& \leq	& (1 + 2\epsilon)s_{k_1}. 
\end{eqnarray*}
Note that the last inequality follows from~\eqref{cepsilonrel}.
\end{proof}

We will now show how to use $P_{k_1}(\hat B)$ to construct an ``approximate cluster,'' (a set with small symmetric difference with one of the clusters) as in Steps~\ref{step:Wj:nonunif}-\ref{step:jstar:nonunif} of Algorithm~\ref{nonunifalg}.  The following lemma gives a way to produce such an approximate cluster using only $\hat B$:

\begin{lemma}\label{largenorm=>goodB}
Assume~\eqref{B-Bhatnorm} holds.  If $|W| \leq (1 + \epsilon)\hat s$ and $||P_{k_1}(\hat B)\1_W||_2 \geq (1 - 3\epsilon)\sqrt s_{k_1}$, then $|W \cap C_i| \geq (1 - 6\epsilon)s_{k_1}$ for some $C_i \in \CC_1$.
\end{lemma}


\begin{proof}
Observe that by~\eqref{shatbounds} we have
\[|W| \leq (1 + \epsilon)(1 + 2\epsilon)s_{k_1} \leq (1 + 4\epsilon)s_{k_1},\]
provided $\epsilon \leq 1 / 2$.  By the triangle inequality, 
\begin{eqnarray}
||P_{k_1}(B)\1_W||_2	& \geq	& ||P_{k_1}(\hat B)\1_W||_2 - ||P_{k_1}(B) - P_{k_1}(\hat B)||_2||\1_W||_2\notag	 \\
& \geq	& (1 - 3\epsilon)\sqrt{s_{k_1}} - \epsilon\sqrt{(1 + 4\epsilon)s_{k_1}}\notag			\\
& \geq	& (1 - 5\epsilon)\sqrt{s_{k_1}}.\label{P1Wtriangle:superclusters}
\end{eqnarray}
We will show that in order for this to hold, $W$ must have large intersection with some cluster in $\CC_1$.

Fix $t$ such that $\frac{(1 + 4\epsilon)s_{k_1}}2 \leq t \leq s_{k_1}$.  Assume by way of contradiction that $|W \cap C_i| \leq t$ for all $i \leq k_1$.  Observe that 
\begin{equation}\label{sumofsquares:superclusters}
||P_{k_1}(B)\1_W||_2^2 = \sum_{i = 1}^k\frac1{s_i}|W \cap C_i|^2.
\end{equation}
Consider the optimization problem
\begin{eqnarray*}
\max			& & \sum_{i = 1}^{k_1}\frac1{s_i}x_i^2	\\
\textrm{s.t.}	& & \sum_{i = 1}^{k_1}x_i \leq (1 + 4\epsilon)s_{k_1},	\\
				& & 0 \leq x_i \leq t \textrm{ for } i = 1, \ldots, k_1,
\end{eqnarray*}
with variable $x_i$ representing $|W \cap C_i|$.  It is easy to see that the maximum occurs when $x_{k_1} = t, x_{k_1 - 1} = (1 + 4\epsilon)s_{k_1} - t$, $x_i = 0$ for all $i  < k_1 - 1$, and the maximum is $\frac{t^2}{s_{k_1}} + \frac{((1 + 4\epsilon)s_{k_1} - t)^2}{s_{k - 1}}$.  Note that the value of of $x_{k_1 - 1}$ is legal by our assumption that $t \geq \frac{(1 + 4\epsilon)s_{k_1}}2$.  Thus, by~(\ref{P1Wtriangle:superclusters}) and~(\ref{sumofsquares:superclusters}) we have
\[(1 - 5\epsilon)^2s_{k_1} \leq ||P_{k_1}(B)\1_W||_2^2 \leq \frac{t^2}{s_{k_1}} + \frac{((1 + 4\epsilon)s_{k_1} - t)^2}{s_{k_1 - 1}} \leq \frac{t^2}{s_{k_1}} + \frac{((1 + 4\epsilon)s_{k_1} - t)^2}{s_{k_1}}.\]
Solving for $t$, this implies
\[t \geq \left(\frac{1 + 4\epsilon}2 + \frac12\sqrt{1 - 28\epsilon + 34\epsilon^2}\right)s_{k_1}.\]
If we make $\epsilon$ small enough ($\epsilon \leq .01$ suffices), then this is $> (1 - 6\epsilon)s_{k_1}$.  Thus, if we pick $t = (1 - 6\epsilon)s_{k_1}$ we have a contradiction.

Therefore, it must be the case that $|W \cap C_i| > (1 - 6\epsilon)s_{k_1}$ for some $i \leq k_1$.  Note that for the proof to go through we require $\frac{1 + 4\epsilon}2 \leq 1 - 6\epsilon$, which is certainly satisfied if $\epsilon \leq .01$.
\end{proof}

This lemma shows that we can a.s.\ produce an approximate cluster by trying all sets $W \subseteq V$ with $|W| \leq (1 + \epsilon)\hat s$ and taking the one which maximizes $||P_{k_1}(\hat B)\1_W||_2$.  However, this would take $\Omega(n^{s_{k_1}})$ time, so we need to narrow the search space.  The next lemma shows that we can, in fact, produce such a $W$ by defining 
\[W_v := \left\{u : \textrm{the $(u, v)$ entry of $P_{k_1}(\hat B)$ is} \geq \frac1{2\hat s}\right\}\]
for each $v \in [n]$ and taking the $W_v$ which maximizes $||P_{k_1}(\hat B)\1_{W_v}||_2$.  This is exactly what is done in Steps~\ref{step:Wj:nonunif}-\ref{step:jstar:nonunif} of Algorithm~\ref{nonunifalg}.

\begin{lemma}\label{existsgoodB}
Assume~\eqref{B-Bhatnorm} holds.  Then there exists $v \in [n]$ such that $|W_v| \leq (1 + \epsilon)\hat s$ and $||P_{k_1}(\hat B)\1_{W_v}||_2 \geq (1 - 3\epsilon)\sqrt{s_{k_1}}$.
\end{lemma}

\begin{proof}
Let $H = (h_{uv})_{u, v = 1}^n$ be the true cluster matrix of $\CC_1$, i.e.
\[h_{uv} := \left\{
\begin{array}{ll}
1	& \textrm{if $u, v \in C$ for some $C \in \CC_1$}	\\
0	& \textrm{else}
\end{array}
\right..\]
This is the matrix that results from rounding the nonzero entries of $P_{k_1}(B)$ to 1.  Now we similarly define $\hat H = (\hat h_{uv})_{u, v = 1}^n$ to be the matrix that results from rounding the ``large'' entries of $P_{k_1}(\hat B)$ to 1:
\[\hat h_{uv} := \left\{
\begin{array}{ll}
1	& \textrm{if the $(u, v)$ entry of $P_{k_1}(\hat B)$ is $\displaystyle\geq \frac1{2\hat s}$}	\\
0	& \textrm{else}
\end{array}
\right..\]
Observe that column $v$ of $\hat H$ is $\1_{W_v}$.

Now consider the errors between $H$ and $\hat H$.  By definition of $\hat H$, each error contributes $\geq \frac1{4\hat s^2}$ to $||P_{k_1}(B) - P_{k_1}(\hat B)||_F^2$.  Thus, by Lemma~\ref{lemma:projl2norm} we have
\[\frac1{4\hat s^2} \cdot \left(
\begin{array}{c}
\textrm{\# errors in}	\\
\textrm{cols.\ } v \in \CC_1
\end{array}
\right) \leq ||P_{k_1}(B) - P_{k_1}(\hat B)||_F^2 \leq 2k_1\epsilon^2.\]
Let $n_1 := s_1 + \ldots + s_{k_1} =$ the number of vertices (columns) in $\CC_1$.  Averaging over the columns in $\CC_1$, there must exist a vertex $v \in \CC_1$ with at most $8k_1\epsilon^2\hat s^2 / n_1$ errors.  Let $C_i$ be the cluster containing $v$.  Then by~\eqref{shatbounds}
\begin{eqnarray*}
|W_v \setminus C_i| + |C_i \setminus W_v|	& =		& \textrm{\# errors in column $v$ of $\hat H$}	\\
	& \leq	& \frac{8k_1\epsilon^2\hat s^2}{n_1}	\\
	& \leq	& \frac{8k_1\epsilon^2(1+ 2\epsilon)^2s_{k_1}^2}{n_1}	\\
	& \leq	& \frac{8n_1\epsilon^2(1 + 2\epsilon)^2s_{k_1}}{n_1}	\\
	& \leq	& 9\epsilon^2s_{k_1}	\\
	& \leq 	& \epsilon s_{k_1}.
\end{eqnarray*}
Thus,
\[|W_v \cap C_i| = s_i - |C_i \setminus W_v| \geq s_i - \epsilon s_{k_1} \geq (1 - \epsilon)s_i,\]
and
\[|W_v| = |W_v \cap C_i| + |W_v \setminus C_i| \leq s_i + \epsilon s_{k_1} \leq (1 + \epsilon)s_i \leq (1 + \epsilon)\hat s.\]
Finally, we must argue that $||P_{k_1}(\hat B)\1_{W_v}||_2 \geq (1 - 3\epsilon)\sqrt{s_{k_1}}$.  First,
\begin{eqnarray*}
||P_{k_1}(B)\1_{W_v}||_2^2	& =		& \sum_{j = 1}^{k_1}\frac{|W_v \cap C_j|^2}{s_j}	\\
	& \geq	& \frac{|W_v \cap C_i|^2}{s_i}	\\
	& \geq	& (1 - \epsilon)^2s_i.
\end{eqnarray*}
Then by the triangle inequality and~\eqref{projl2norm}
\begin{eqnarray*}
||P_{k_1}(\hat B)\1_{W_v}||_2	& \geq	& ||P_{k_1}(B)\1_{W_v}||_2 - ||P_{k_1}(B) - P_{k_1}(\hat B)||_2||\1_{W_v}||_2	\\
						& \geq	& (1 - \epsilon)\sqrt{s_i} - \epsilon\sqrt{|W|}	\\
						& \geq	& (1 - \epsilon)\sqrt{s_i} - \epsilon\sqrt{(1 + \epsilon)s_i}	\\
						& \geq 	& (1 - 3\epsilon)\sqrt{s_i}.
\end{eqnarray*}
This completes the proof.
\end{proof}

Lemmas~\ref{largenorm=>goodB} and~\ref{existsgoodB} fit together as follows: Lemma~\ref{largenorm=>goodB} shows that any set $W$ such that $W \leq (1 + \epsilon)\hat s$ and $||P_{k_1}(\hat B)\1_W||_2 \geq (1 - 3\epsilon)\sqrt{s_{k_1}}$ must come mostly from a single cluster, while Lemma~\ref{existsgoodB} shows that there must be such a $W$ among the $W_v$.  Thus, we can a.s.\ produce an approximate cluster $W$ by simply taking the $W_v$ such that $|W_v| \leq (1 + \epsilon)\hat s$ and $||P_{k_1}(\hat B)\1_{W_v}||_2$ is maximum. 

Note that this approach does not require any access to $s_1, \ldots, s_m$.  However, we assume $k_1, \ldots, k_K$ are known so that we know how many eigenvectors to project onto (i.e., $k = k_1$).

\subsection{Recovering the cluster exactly}\label{exactcluster}

Once we have a set $W$ with small symmetric difference with a cluster $C_i$, we show how to recover $C_i$ \emph{exactly} a.s.\ by looking at the number of neighbors each vertex has in $W$.  First, we show that vertices in $C_i$ are distinguished from those outside $C_i$ by their number of neighbors in $C_i$ itself (Lemma~\ref{recoverclusterswhp}).  Then we show that using $W$ in place of $C_i$ does not throw things off by too much (Lemma~\ref{recoverexactly}).

\begin{lemma}\label{recoverclusterswhp}
Consider cluster $C_i$ and vertex $u \in [n]$.  If $u \in C_i$, then
\begin{equation}\label{jinCi}
|N_{\hat G}(u) \cap C_i| \geq (p - \epsilon)s_i
\end{equation}
with probability $\geq 1 - e^{-\epsilon^2s_i}$, and if $u \notin C_i$, then
\begin{equation}\label{jnotinCi}
|N_{\hat G}(u) \cap C_i| \leq (q + \epsilon)s_i
\end{equation}
with probability $\geq 1 - e^{-\epsilon^2s_i}$.
\end{lemma}

The proof is essentially the same as that of~\cite[Lemma~19]{ColeFR17} and is therefore left as an exercise.

\begin{lemma}\label{recoverexactly}
Assume~\eqref{B-Bhatnorm} holds.  Suppose $|W| \leq (1 + \epsilon)\hat s$ and $|W \cap C_i| \geq (1 - 6\epsilon)s_{k_1}$ for some $C_i \in \CC_1$.  Then 
\begin{enumerate}[a)]
\item	If $u \in C_i$ and $u$ satisfies~(\ref{jinCi}), then $|N_{\hat G}(u) \cap W| \geq (p - 10\epsilon)\hat s$.\label{recoverclustera:agnostic}
\item	If $u \in [n] \setminus C_i$ and $u$ satisfies~(\ref{jnotinCi}), then $|N_{\hat G}(u) \cap W| \leq (q + 10\epsilon)\hat s$.\label{recoverclusterb:agnostic}
\end{enumerate}
\end{lemma}

\begin{proof}
Assume $u \in C_i$ and $u$ satisfies~(\ref{jinCi}).  We want to lower bound $|N(u) \cap W|$ in terms of $|N(u) \cap C_i|$.  The worst case is when as many as possible of $u$'s neighbors in $C_i$ come from $C_i \setminus W$, i.e., when $u$ is adjacent to all vertices in $C_i \setminus W$.  Thus, we have
\[|N(u) \cap W| \geq |N(u) \cap C_i| - |C_i \setminus W| \geq (p - \epsilon)s_i - |C_i \setminus W|.\]
As $|C_i| = s_i$ and $|W \cap C_i| \geq (1 - 6\epsilon)s_{k_1}$, we have 
\[|C_i \setminus W| \leq s_i - (1 - 6\epsilon)s_{k_1} \leq s_i - \frac{1 - 6\epsilon}{1 + \epsilon}s_i \leq 7\epsilon s_i.\]
Therefore,
\[|N(u) \cap W| \geq (p - \epsilon)s_i - 7\epsilon s_i \geq (p - 8\epsilon)s_{k_1}.\]
Therefore, by~\eqref{shatbounds} we have
\[|N(u) \cap W| \geq \frac{p - 8\epsilon}{1 + 2\epsilon}\hat s \geq (p - 10\epsilon)\hat s.\]
This proves part~\ref{recoverclustera:agnostic}).

For part~\ref{recoverclusterb:agnostic}), assume $u \notin C_i$ and $u$ satisfies~(\ref{jnotinCi}).  Now we want to upper bound $|N(u) \cap W|$ in terms of $|N(u) \cap C_i|$.  Now the worst case is when $u$ has as many neighbors as possible in $W \setminus C_i$, i.e., when $u$ is adjacent to all vertices in $W \setminus C_i$.  In this case,
\[|N(u) \cap W| \leq |N(u) \cap C_i| + |W \setminus C_i| \leq (q + \epsilon)s_i + |W \setminus C_i|.\]
As $|W| \leq (1 + \epsilon)\hat s$ and $|W \cap C_i| \geq (1 - 6\epsilon)s_{k_1}$, we have 
\[|W \setminus C_i| \leq (1 + \epsilon)\hat s - (1 - 6\epsilon)s_{k_1} \leq (1 + \epsilon)\hat s - \frac{1 - 6\epsilon}{1 + 2\epsilon}\hat s \leq 9\epsilon\hat s.\]
Therefore,
\[|N(u) \cap W| \leq (q + \epsilon)s_i + 9\epsilon\hat s \leq (q + 10\epsilon)\hat s.\]
This completes the proof of~\ref{recoverclusterb:agnostic}).
\end{proof}

Thus, if we have a set $W$ which has large intersection with $C_i$, we can use $|N(u) \cap W|$ to distinguish between $u \in C_i$ and $u \notin C_i$ as shown in Figure~\ref{fig:recovercluster}, provided 
\[p - 10\epsilon > q + 10\epsilon,\]
or, equivalently,
\begin{equation*}\label{epsilonpqrel}
\epsilon < \frac{p - q}{20}.
\end{equation*}

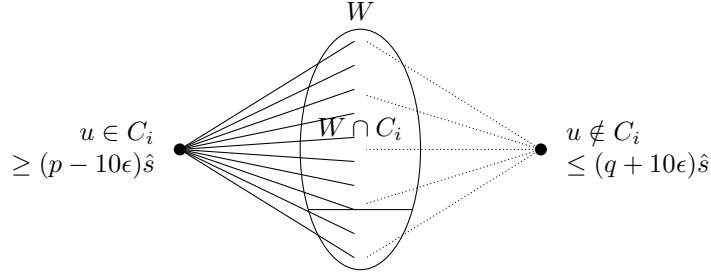
\begin{figure}
\centering
\begin{tikzpicture}[scale=.08]
\draw (0, 0) ellipse(10 and 20);
\draw (-8.66, -10) -- (8.66, -10);
\draw (0, 0) node[anchor=south]{$W \cap C_i$};
\draw (0, 20) node[anchor=south]{$W$};
\foreach \x in {-30, 30}
\fill (\x, 0) circle(1);
\draw (-30, 0) node[anchor=east]{\begin{tabular}{r}$u \in C_i$ \\ $\geq (p - 10\epsilon)\hat s$\end{tabular}};
\draw (30, 0) node[anchor=west]{\begin{tabular}{l}$u \notin C_i$ \\ $\leq (q + 10\epsilon)\hat s$\end{tabular}};
\foreach \y in {-18, -14, ..., 18}
\draw (-30, 0) -- (-1, \y);
\foreach \y in {-18, -9, 0, 9, 18}
\draw[densely dotted] (30, 0) -- (1, \y);
\end{tikzpicture}
\caption{If $W$ has large overlap with $C_i$, then a.s.\ vertices in $C_i$ will have many neighbors in $W$, while vertices not in $C_i$ will have relatively few neighbors in $W$.}\label{fig:recovercluster}
\end{figure}

Note that we apply Lemma~\ref{recoverexactly} to the $W_u$, which themselves depend on the random sample $\hat G$, so we cannot simply treat $|N(u) \cap W|$ as the sum of $|W|$ independent random variables and follow a Hoeffding argument as in Lemma~\ref{recoverclusterswhp}.  This is why we need both Lemmas~\ref{recoverclusterswhp} and~\ref{recoverexactly}.

\section{The ``delete and recurse'' step}\label{delandrec}

After we have found one cluster, we cannot simply say that Algorithm~\ref{nonunifalg} finds the remaining clusters by the same argument.  Some care has to be taken because the iterations of Algorithm~\ref{nonunifalg} cannot be handled independently: the event that iteration $t$ correctly recovers a cluster certainly depends on whether or not iterations $1, \ldots, t - 1$ correctly recovered clusters.  

We can get around this by ``preprocessing the randomness'' as in~\cite[Section~7.3]{ColeFR17}.  Essentially, we apply the analysis in Sections~\ref{BBhateigs}-\ref{recsingle} to all $2^k$ \emph{cluster submatrices} of $\hat B$ (principle submatrices induced by a subset of the clusters) and show via a union bound that the overall failure probability is still small.

Formally, we define the $2^k$ cluster submatrices as
\[\hat B_J := \hat B\left[\bigcup_{i \in J}C_i\right]\]
for $J \subseteq [k]$.  We define cluster submatrices of $B$ analogously  Next, we define the following events on $\mathcal G(n, \CC, p, q)$:
\begin{itemize}
\item	\emph{Spectral events}: for $J \subseteq [k]$, let $E_J$ be the event that $||B_J - \hat B_J||_2 \leq 7\sqrt{\matrdim(B_J)} = 7\sqrt{\sum_{i \in J}s_i}$.  These are the events that the eigenvalues of the cluster submatrices are close to their expectations.
\item	\emph{Degree events}: for $1 \leq i \leq k, 1 \leq u \leq n$, let $D_{i, u}$ be the event that $|N_{\hat G}(u) \cap C_i| \geq (p - \epsilon)s_i$ if $u \in C_i$, or the event that $|N_{\hat G}(u) \cap C_i| \leq (q + \epsilon)s_i$ if $u \notin C_i$.  These are the events that each vertex has approximately the expected number of neighbors in each cluster.
\end{itemize}
Thus, we have defined a total of $2^k + nk$ events.  Essentially, these are the events that every $\hat B_J$ satisfies~\eqref{B-Bhatnorm} and that~(\ref{jinCi}) and~(\ref{jnotinCi}) are satisfied for all $i \in [k], u \in [n]$.  Note that the events are well-defined, as their definitions depend only on the underlying probability space $\mathcal G(n, \mathcal C, p, q)$ and not on the random graph $\hat G$ sampled from the space.

These final two lemmas prove that Algorithm~\ref{nonunifalg} succeeds with probability at least
\[1 - \left(\frac2e\right)^{s_k} + \frac{n^{3 / 2}}{e^{\epsilon^2s_k}} = 1 - 2^{-\Omega(\sqrt n)}:\]

\begin{lemma}
Assume $E_J$ and $D_{i, u}$ hold for all $J \subseteq [k]$, $1 \leq i \leq k$, and $1 \leq u \leq n$.  Then Algorithm~\ref{nonunifalg} successfully recovers $\CC$.
\end{lemma}

\begin{lemma}
$\displaystyle \prb\left[\bigcup_{J \subseteq [k]}\bar E_J \cup \bigcup_{i \in [k], j \in [n]} \bar D_{i, u}\right] \leq \left(\frac2e\right)^{s_k} + \frac{n^{3 / 2}}{e^{\epsilon^2s_k}} = 2^{-\Omega(\sqrt n)}$.
\end{lemma}

We omit the proofs, as they are essentially the same as the proof of the main theorem in~\cite[Section~7.3]{ColeFR17}.  The main difference is that one argues that Algorithm~\ref{nonunifalg} first recovers $\CC_1$, then $\CC_2$, etc.  Thus, we really only have to take the union bound over $2^{k_1} + \ldots + 2^{k_K}$ cluster submatrices, not all $2^k$.


\section{Parameter-free planted partition}\label{paramfree}

Until this point, we have assumed that our algorithm has access to $p$, $q$, and $k_1, \ldots, k_K$, but not $s_1, \ldots, s_k$.  In this section, we discuss what to do when we don't have access to these parameters' exact values.

\subsection{Unknown supercluster sizes}\label{unknownki}

Let us assume that \emph{only} $p$ and $q$ are known.  As it turns out, we can reduce this case to the case when $k_1, \ldots, k_K$ are known at the expense of slightly increasing $c$.

Assume~\eqref{B-Bhatnorm} holds.  Then by~\eqref{superclustersep}
\begin{equation}\label{superclustersepBhat}
\begin{array}{l}
C_i, C_{i + 1} \textrm{ in different}	\\
\textrm{superclusters}
\end{array}
\Rightarrow 
\lambda_i(\hat B) \geq \lambda_{i + 1}(\hat B) + ((p - q)c - 14)\sqrt n.
\end{equation}
Thus, let us go down the list of eigenvalues of $\hat B$ in decreasing order and whenever we see a separation of at least $((p - q)c - 14)\sqrt n$ record the number of eigenvalues seen since the last such separation.  Ignore the last group, as these eigenvalues correspond to the $n - k$ zero eigenvalues of $B$.  Let $\hat k_1, \ldots, \hat k_L$ be this sequence of numbers.  

Consider the clustering
\begin{eqnarray*}
\hat\CC_1 	& :=	& \{C_1, \ldots, C_{\hat k_1}\}, \\
\hat\CC_2 	& :=	& \{C_{\hat k_1 + 1}, \ldots, C_{\hat k_1 + \hat k_2}\}, 	\\
			& \vdots	\\
\hat\CC_L 	& :=	& \{C_{\hat k_1 + \ldots + \hat k_{L - 1} + 1}, \ldots, C_{\hat k_1 + \ldots + \hat k_L}\}.
\end{eqnarray*}

\begin{obs}\label{refinement}
If~\eqref{B-Bhatnorm} holds, then $\hat\CC_1, \ldots, \hat\CC_L$ is a refinement of $\CC_1, \ldots, \CC_K$ (as partitions of $\CC$).
\end{obs}

\begin{proof}
Consider $\hat\CC_1$.  By definition of $\hat k_1$, 
\[\lambda_i(\hat B) < \lambda_{i + 1}(\hat B) - ((p - q)c - 14)\sqrt n\]
for $i = 1, \ldots, \hat k_1 - 1$.  By the contrapositive of~\eqref{superclustersepBhat}, this means $C_i, C_{i + 1}$ are in the same $\CC_j$ for $i = 1, \ldots, \hat k_1 - 1$.  A similar argument applies for each $\hat\CC_i$.  Hence, each $\hat\CC_i$ is a subset of some $\CC_j$. 
\end{proof}

\begin{obs}\label{refinementsep}
Assume~\eqref{B-Bhatnorm} holds.  Then $s_i \geq s_{i + 1} + \left(c - \frac{28}{p - q}\right)\sqrt n$ whenever $C_i, C_{i + 1}$ are in different $\hat\CC_j$.  Equivalently,
\[\min_{C \in \hat\CC_j}|C| \geq \max_{C \in \hat\CC_{j + 1}}|C| + \left(c - \frac{28}{p - q}\right)\sqrt n\]
for $i = 1, \ldots, L - 1$.
\end{obs}

Thus, by these two observations, we a.s.\ have a separation into superclusters $\hat\CC_1, \ldots, \hat\CC_L$ such that
\begin{itemize}
\item	$s_i \geq s_{i + 1} + c'\sqrt n$ whenever $C_i, C_{i + 1}$ are in different $\hat\CC_j$, where
\begin{equation}\label{c'def}
c' := c - \frac{28}{p - q}
\end{equation}
(by Observation~\ref{refinementsep}).
\item	$s_i \leq (1 + \epsilon)s_{i + 1}$ whenever $C_i, C_{i + 1}$ are in the same $\hat\CC_j$ (by Observation~\ref{refinement} and~\eqref{withinsupercluster}).
\end{itemize}
This is sufficient to recover $\CC$, since we only care about recovering the individual clusters, not which supercluster each cluster belongs to.

Note that since $c = \Omega((p - q)^{-2})$, $c'$ is still a large positive constant.  Hence, in order for the analysis in Sections~\ref{BBhateigs}-\ref{delandrec} to go through, we simply require
\[c \geq \frac{14}{(p - q)\epsilon} + \frac{28}{p - q} = \frac{14}{(p - q)\epsilon}(1 + 2\epsilon).\]
in place of~\eqref{cepsilonrel}.  As $\epsilon \leq .01$, setting $c \geq \frac{15}{(p - q)\epsilon}$ suffices.

\subsection{Unknown $p$ and $q$}\label{unknownpq}

Unfortunately, we still need to assume that Algorithm~\ref{nonunifalg} has access to the exact values of $p$ and $q$ so that it can compute $\hat B := \hat A + pI_n - qJ_n$.  There are two possible ways around this:
\begin{enumerate}
\item	Obtain good estimates on the eigenvalues of $A$ so that the algorithm can use $\hat A$ instead of $\hat B$, as in the uniform case.  Since
\begin{equation}\label{rank1pert}
A = B + qJ_n,
\end{equation}
Weyl's inequalities give the bounds
\[(p - q)s_i \leq \lambda_i(A) \leq (p - q)s_{i - 1}\]
for $i = 2, \ldots, k$.  However, in order for the spectral results in Sections~\ref{BBhateigs} and~\ref{projdev} to go through we need a separation between $\lambda_i(A)$ and $\lambda_{i - 1}(A)$ when $C_i$ and $C_{i - 1}$ are in different superclusters.  Thus, the above bounds are not good enough.  However, by~\eqref{rank1pert} we may view $A$ as a rank-1 perturbation of $B$ (whose eigenvalues are known), so we may attempt to use perturbation results such as~\cite{Ding2007, gu1994stable, Mehl2011} to compute its eigenvalues.
\item	Estimate $p$ and $q$ empirically.  We must find a way to do so with $\leq O(1 / \sqrt n)$ error in order to overcome the $O(\sqrt n)$ error introduced by the random noise (see Lemma~\ref{lemma:B-Bhatnorm}).    One promising approach is to use the techniques of \emph{graphons}~\cite{lovasz2012large}.  In~\cite{kenyon2017multipodal, kenyon2017phases}, the authors use the theory of large deviations~\cite{chatterjee2016introduction} to show that the edge and triangle or $l$-star densities of ``most'' graphs together induce a multipodal (i.e.\ stochastic block model) structure in the limiting graphon.  Thus, one might hope to estimate $p$, $q$, and $s_1, \ldots, s_k$ by looking at the graphon induced by these statistics on $\hat G$.
\end{enumerate}.

\section{Planted partitions in random symmetric matrices}\label{gen}

We now attempt to push Algorithm~\ref{nonunifalg} to an even more general setting.  In the above sections we receive as input a random graph or, equivalently, a random symmetric matrix $\hat A = (\hat a_{uv})_{u, v = 1}^n$ whose diagonal entries are 0 and whose off-diagonal entries are Bernoulli random variables with expectation $p$ or $q$.  More generally, we can assume that $\hat A$ is a random symmetric matrix whose entries come from arbitrary distributions (under certain assumptions) with expectations $p$ and $q$.

\begin{defn}[Planted partition model for random symmetric matrices]
Let $\mathcal C = \{C_1, \ldots, C_k\}$ be a partition of the set $[n]$ into $k$ clusters.  For distributions $D_1, D_2, D_3$ on $\R$, we define the \emph{planted partition model} $\mathcal{PP}(n, \mathcal C, D_1, D_2, D_3)$ to be the probability space of real symmetric $n \times n$ matrices  $\hat A = (\hat a_{uv})_{u, v = 1}^n$, where $\hat a_{uv}$ are distributed independently for $1 \leq u \leq v \leq n$ such that
\[\hat a_{uv} \sim \left\{
\begin{array}{ll}
D_1	& \textrm{if $u \neq v$ and $u, v$ in the same cluster}	\\
D_2	& \textrm{if $u \neq v$ and $u, v$ in different clusters}	\\
D_3	& \textrm{if $u = v$}
\end{array}
\right..\]
Furthermore, we assume the following:
\begin{enumerate}
\item	$\expval[D_1] = p, \expval[D_2] = q$, and $\expval[D_3] = 0$, where $0 \leq q < p$.
\item	$\var[\hat a_{uv}] \leq \sigma^2$ for all $u, v$.
\item	$|\hat a_{uv} - \expval[\hat a_{uv}]| \leq \kappa$ for all $u, v$.  I.e., the support of each random variable $\hat a_{uv}$ is contained in an interval of length $2\kappa$ centered at its mean.
\end{enumerate}
\end{defn}

\begin{prob}[Planted partition in a random symmetric matrix]
Identify (or ``recover'') the unknown partition $C_1, \ldots, C_k$ (up to a permutation of $[k]$) given only a random matrix $\hat A \sim \mathcal{PP}(n, \mathcal C, D_1, D_2, D_3)$.
\end{prob}

We will assume that $\CC$ satisfies the superclusters assumptions of Section~\ref{superclusters}, with the following changes:
\begin{itemize}
\item 	We replace~\eqref{siassumption} and~\eqref{superclustersep} with 
\[s_1 \geq \ldots \geq s_k \geq \Delta\]
and
\begin{equation*}\label{superclustersepgen}
\min_{C \in \CC_i} |C| \geq \max_{C \in \CC_{i + 1}}|C| + \Delta,
\end{equation*}
respectively. 
\item 	We now allow the parameters $p, q, \sigma, \kappa, \Delta$, and $\epsilon$ to depend on $n$.  
\item	We assume without loss of generality that $p > q$, but we no longer require that $0 \leq q < p \leq 1$.  
\end{itemize}
Our goal is thus to give conditions on the various parameters that are sufficient for Algorithm~\ref{nonunifalg} to succeed a.s.\ (Section~\ref{paramdep}), noting that the algorithm now receives a random matrix $\hat A \sim \PP(n, \CC, D_1, D_2, D_3)$ as input instead of a random graph $\hat G \sim \mathcal G(n, \CC, p, q)$.  

We now indicate briefly the changes to the analysis in Sections~\ref{BBhateigs}-\ref{delandrec} necessary to make Algorithm~\ref{nonunifalg} work in this more general setting, leaving the details as an exercise.

\subsection{Spectral results}

We can define $A$, $B$, and $\hat B$ as in Section~\ref{notation}.  The main difference in the spectral results of Section~\ref{BBhateigs} is that we get
\begin{equation}\label{nonbernoullinorm}
||B - \hat B||_2 \leq (2\sigma + 6\kappa)\sqrt n
\end{equation}
a.s.\ in place of Lemma~\ref{B-Bhatnorm}.  This dependence on $\sigma$ and $\kappa$ makes sense intuitively because these parameters control how concentrated the entries of $\hat B$ are about their means: if they are too spread out, we should not expect $B$ and $\hat B$ to be ``close.''  Conversely, if $\sigma, \kappa = o(1)$, then we should expect $\hat B$ to be closer to $B$ than when these parameters are constant.

This yields a separation of $\geq (p - q)\Delta - (4\sigma + 12\kappa)\sqrt n$ between the eigenvalues of $B$ and $\hat B$ corresponding to indices in different superclusters (cf.\ Lemma~\ref{eigsep}).  This, in turn, allows us to bound $P_{k_1}(B) - P_{k_1}(\hat B)$ in norm as in Lemma~\ref{lemma:projl2norm}, provided that
\begin{equation}\label{nonbernoulliepsilondep}
\Delta \geq \frac{(4\sigma + 12\kappa)\sqrt n}{(p - q)\epsilon}.
\end{equation}

\subsection{Constructing an approximate cluster}

In Step~\ref{step:shat} of Algorithm~\ref{nonunifalg}, define
\[\hat s := \frac{\lambda_1(\hat B) + (2\sigma + 6\kappa)\sqrt n}{p - q}.\]
Then~(\ref{shatbounds}) holds as in the Bernoulli case, assuming~(\ref{nonbernoullinorm}) holds.  Thus, Lemmas~\ref{largenorm=>goodB} and~\ref{existsgoodB} remain exactly the same as in the Bernoulli case (except replace~(\ref{B-Bhatnorm}) with~(\ref{nonbernoullinorm})).

\subsection{Recovering the cluster exactly}

In the Bernoulli case, we use the random variables
\[|N_{\hat G}(u) \cap C_i| = \sum_{v \in C_i}\hat a_{uv}\]
to distinguish between $u \in C_i$ and $u \notin C_i$ (Lemmas~\ref{recoverclusterswhp} and~\ref{recoverexactly}).  Thus, for general distributions we define the random variable
\[S_{u, W} := \sum_{v \in W}\hat a_{uv}\]
for $u \in [n]$ and $W \subseteq [n]$.  (Recall that a random variable is actually a measurable function from a probability space to $\R$; hence, $S_{u, W}$ is actually a function of the random matrix $\hat A$.)

Proceeding as in Lemma~\ref{recoverclusterswhp}, for each cluster $C_i$ we get 
\begin{equation}\label{SuCiinCi}
S_{u, C_i} \geq (p - \epsilon)s_i
\end{equation}
for all $u \in C_i$, and
\begin{equation}\label{SuCinotinCi}
S_{u, C_i} \leq (q + \epsilon)s_i
\end{equation}
for all $u \notin C_i$, each with probability $\displaystyle\geq 1 - \exp\left(-\frac{\epsilon^2s_i}{3\kappa^2}\right)$.

We now argue that if $W$ has large intersection with some $C_i$, then we get bounds on $S_{u, W}$ which are not far off from those on $S_{u, C_i}$.  However, since the entries of $\hat A$ need not be 0 or 1, each element of $W \triangle C_i$ can throw off the bounds by as much as
\begin{equation*}
\max_{u, v}|\hat a_{uv}| \leq \mu + \kappa,
\end{equation*}
where
\begin{equation*}\label{mudef}
\mu := \max\{|p|, |q|\}.
\end{equation*}
More precisely:

\begin{lemma}\label{recoverexactly:nonber}
Assume~\eqref{nonbernoullinorm} holds.  Suppose $|W| \leq (1 + \epsilon)\hat s$ and $|W \cap C_i| \geq (1 - 6\epsilon)s_k$ for some $C_i \in \CC_1$.  Then 
\begin{enumerate}[a)]
\item	If $u \in C_i$ and $u$ satisfies~(\ref{SuCiinCi}), then $S_{u, W} \geq (p - (18\mu + 16\kappa + 1)\epsilon)\hat s$.
\item	If $u \in [n] \setminus C_i$ and $u$ satisfies~(\ref{SuCinotinCi}), then $S_{u, W} \leq (q + (16\mu + 16\kappa + 1)\epsilon)\hat s$.
\end{enumerate}
\end{lemma}
\noindent	We omit the proof, as it parallels that of Lemma~\ref{recoverexactly}.

Thus, we are able to recover the cluster a.s.\ as long as
\[p - (18\mu + 16\kappa + 1)\epsilon > q + (16\mu + 16\epsilon + 1)\epsilon,\]
or equivalently,
\begin{equation}\label{genepsilon}
\epsilon < \frac{p - q}{34\mu + 32\kappa + 2}.
\end{equation}
It may be possible to optimize these constants slightly by breaking the proof of Lemma~\ref{recoverexactly:nonber} into cases based on whether $p$ and $q$ are positive or negative, but there will always be a dependence on $p, q$ and $\kappa$.

\subsection{Parameter dependencies}\label{paramdep}

By~\eqref{nonbernoulliepsilondep}, the above inequality~\eqref{genepsilon} is satisfied if
\[\frac{(2\sigma + 6\kappa)\sqrt n}{(p - q)\Delta} \leq \frac{p - q}{34\mu + 32\kappa + 2}.\]
This can be accomplished if we require
\begin{equation}\label{finaldep}
\Delta = \Omega\left(\frac{\kappa(\max\{|p|, |q|\} + \kappa)\sqrt n}{(p - q)^2}\right).
\end{equation}
In addition, observe that we need the failure probability of $\exp\left(-\frac{\epsilon^2s_i}{3\kappa^2}\right)$ above to be $o(nk)$, since we take a union bound over $nk$ ``degree events'' (see Section~\ref{delandrec}).  This can be accomplished if we require
\begin{equation}\label{finaldep2}
\Delta = \omega\left(\frac{(\max\{|p|, |q|\} + \kappa)^2\kappa^2\log n}{(p - q)^2}\right).
\end{equation}
Thus, we get the following theorem:

\begin{thm}\label{genthm}
Let $\mathcal C$ be defined as in the beginning of Section~\ref{gen}, and assume that~\eqref{genepsilon}-\eqref{finaldep2} are satisfied.  Then $\mathcal C$ can be recovered a.s.\ in polynomial time given only $\hat A \sim \PP(n, \mathcal C, D_1, D_2, D_3)$.
\end{thm}

Observe that~\eqref{finaldep} and~\eqref{finaldep2} together imply
\[\kappa(\max\{|p|, |q|\} + \kappa\}) \ll \frac{\sqrt n}{\log n}.\]
So a necessary condition for Algorithm~\ref{nonunifalg}'s success is that $|p|, |q|$, and $\kappa$ aren't too big.  On the other hand, if $\kappa$ is small (i.e., the entries of $\hat B$ are highly concentrated), we can potentially get away with a smaller-than-$\sqrt n$ separation between the cluster sizes.  Note also that our goal was to make Algorithm~\ref{nonunifalg} work in the \emph{most general} setting possible; it may be possible to obtain better conditions on the parameters in certain special cases.  One can do this by mirroring the analysis in Sections~\ref{BBhateigs}-\ref{delandrec}.

\bibliographystyle{plain}
\bibliography{bibliography}

\begin{thebibliography}{10}

\bibitem{chatterjee2016introduction}
Sourav Chatterjee.
\newblock An introduction to large deviations for random graphs.
\newblock {\em Bulletin of the American Mathematical Society}, 53(4):617--642,
  2016.

\bibitem{ColeFR17}
Sam Cole, Shmuel Friedland, and Lev Reyzin.
\newblock A simple spectral algorithm for recovering planted partitions.
\newblock {\em Special Matrices}, 5(1):139--157, 2017.

\bibitem{Ding2007}
Jiu Ding and Aihui Zhou.
\newblock Eigenvalues of rank-one updated matrices with some applications.
\newblock {\em Applied Mathematics Letters}, 20(12):1223 -- 1226, 2007.

\bibitem{er1959}
Paul Erd{\H o}s and Alfr{\'e}d R{\'e}nyi.
\newblock On random graphs {I}.
\newblock {\em Publicationes Mathematicae (Debrecen)}, 6:290--297, 1959 1959.

\bibitem{friedland2015matrices}
Shmuel Friedland.
\newblock {\em Matrices}.
\newblock World Scientific, 2015.

\bibitem{FK81}
Zolt{\'{a}}n F{\"{u}}redi and J{\'{a}}nos Koml{\'{o}}s.
\newblock The eigenvalues of random symmetric matrices.
\newblock {\em Combinatorica}, 1(3):233--241, 1981.

\bibitem{gu1994stable}
Ming Gu and Stanley~C Eisenstat.
\newblock A stable and efficient algorithm for the rank-one modification of the
  symmetric eigenproblem.
\newblock {\em SIAM journal on Matrix Analysis and Applications},
  15(4):1266--1276, 1994.

\bibitem{kenyon2017multipodal}
Richard Kenyon, Charles Radin, Kui Ren, and Lorenzo Sadun.
\newblock Multipodal structure and phase transitions in large constrained
  graphs.
\newblock {\em Journal of Statistical Physics}, 168(2):233--258, 2017.

\bibitem{kenyon2017phases}
Richard Kenyon, Charles Radin, Kui Ren, and Lorenzo Sadun.
\newblock The phases of large networks with edge and triangle constraints.
\newblock {\em arXiv preprint arXiv:1701.04444}, 2017.

\bibitem{lovasz2012large}
L{\'a}szl{\'o} Lov{\'a}sz.
\newblock {\em Large networks and graph limits}, volume~60.
\newblock 2012.

\bibitem{Mehl2011}
Christian Mehl, Volker Mehrmann, André~C.M. Ran, and Leiba Rodman.
\newblock Eigenvalue perturbation theory of classes of structured matrices
  under generic structured rank one perturbations.
\newblock {\em Linear Algebra and its Applications}, 435(3):687 -- 716, 2011.
\newblock Special Issue: Dedication to Pete Stewart on the occasion of his 70th
  birthday.

\end{thebibliography}

\appendix

\section{Proof of Lemma~\ref{projdiff}}\label{projdiffproof}

We will prove~\eqref{projdiffl2} using the Cauchy integral formula for projections~\cite[Theorem~14]{ColeFR17}.  Define $\gamma$ to be the boundary of a $2M \times 2M$ square in the complex plane with upper and lower sides are on the lines $y = \pm M$, left side on the line $x = x_0$, and right side on the line $x = x_0 + 2M$, where $x_0 := (\alpha + \beta) / 2$ and we will let $M \to \infty$.  Thus, the interior of $\gamma$ contains exactly the largest $r$ eigenvalues of both $X$ and $Y$ (see Figure~\ref{cauchypic}).  

Applying the Cauchy integral formula,
\begin{eqnarray*}\label{Caucintfor}
P_r(X)=\frac{1}{2\pi{i}}\int_{\gamma} (zI_n-X)^{-1}dz,\\ 
P_r(Y)=\frac{1}{2\pi{i}}\int_{\gamma} (zI_n-Y)^{-1}dz.\notag
\end{eqnarray*}
Hence
\begin{eqnarray*}
P_r(X) - P_r(Y)&=&\frac{1}{2\pi{i}}\int_{\gamma} (zI_n-X)^{-1}\big((zI_n-Y)-(zI_n-X)\big)(zI_n-Y)^{-1} dz\\
&=&\frac{1}{2\pi{i}}\int_{\gamma} (zI_n-X)^{-1}\big(X-Y\big)(zI_n-Y)^{-1} dz,
\end{eqnarray*}
and so we get
\begin{eqnarray}
\|P_r(X) - P_r(Y)\|_2&\le& \frac{1}{2\pi}\int_{\gamma} \|(zI_n-X)^{-1}\big(X-Y\big)(zI_n-Y)^{-1}\| _2|dz| \label{Cauchyforest} \\
&\le& \frac{1}{2\pi}\int_{\gamma} \|(zI_n-X)^{-1}\|_2\|X-Y\|_2 \|(zI_n-Y)^{-1}\|_2 |dz|.\notag
\end{eqnarray}
Observe that for each $z\in\C$ the matrices $zI_n-X, zI_n-Y $ are normal.
Hence 
\[\|(zI_n-X)^{-1}\|_2=\frac{1}{\min_{j\in [n]} |z-\lambda_j(X)|}, \quad  \|(zI_n-Y)^{-1}\|_2=\frac{1}{\min_{j\in [n]} |z-\lambda_j(Y)|}.\]
Let us first estimate the contribution to the integral \eqref{Cauchyforest} on the left side of $\gamma$.  Let $z=x_0 +y{i}, y\in\R$.  That is, $z$ lies on the line $x=x_0$.
By our assumption about the eigenvalues of $X$ and $Y$, all eigenvalues of $X$ and $Y$ are a horizontal distance of at least $d := (\beta - \alpha) / 2$ from the line $x = x_0$; hence
\[|z-\lambda_j(X)|, |z - \lambda_j(Y)| \ge \sqrt {d^2 +y^2}\]
for $z = x_0 + yi$.  Therefore, the contribution to~\eqref{Cauchyforest} from the left side of $\gamma$ is upper bounded by
\[\frac1{2\pi}\int_{-\infty}^\infty\frac{||X - Y||_2}{d^2 + y^2}dy = \frac{||X - Y||_2}{2d} = \frac{||X - Y||_2}{\beta - \alpha}.\]
The contributions from the other sides of $\gamma$ go to 0 as $M \to \infty$.  This completes the proof of~\eqref{projdiffl2}.

To show~\eqref{projdiffFrob}, observe that $P_r(X)$ and $P_r(Y)$ both have rank $r$, so $P_r(X)-P_r(Y)$ has rank at most $2r$.  Hence, $P_r(X)-P_r(Y)$ has at most $2r$ nonzero eigenvalues.  Recall that for any real symmetric $n \times n$ matrix $H$ 
\[||H||_F^2 = \sum_{i = 1}^n\lambda_i(H)^2 \leq \rk(H) \cdot ||H||_2^2.\]
The lemma thus follows from~\eqref{projdiffl2}:
\[||P_r(X) - P_r(Y)||_F^2 \leq 2r||P_r(X) - P_r(Y)||_2^2 \leq \frac{2r||X - Y||_2^2}{(\beta - \alpha)^2}.\qed\]

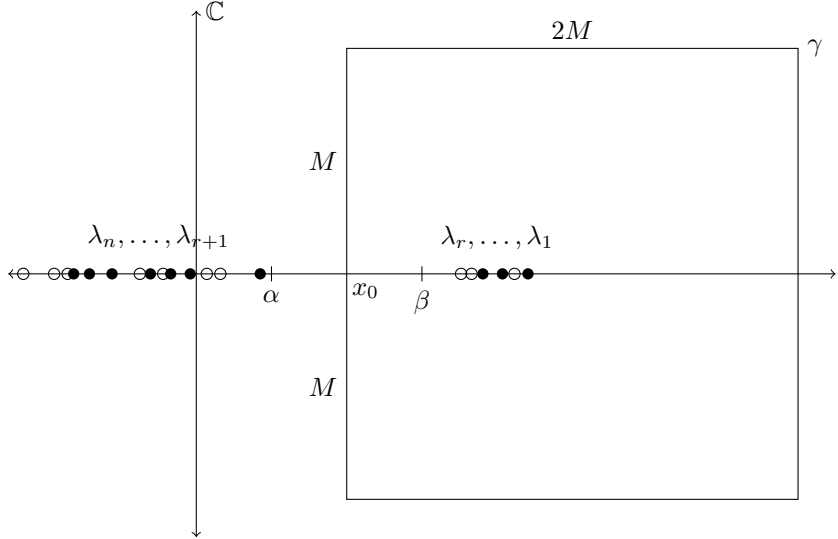
\begin{figure}
\centering
\begin{tikzpicture}

\draw[<->] (0, -3.5) -- (0, 3.5) node[anchor = west]{$\C$};
\draw[<->] (-2.5, 0) -- (8.5, 0);

\draw (2, -3) rectangle(8, 3) node[anchor = west]{$\gamma$};
\draw (5, 3) node[anchor = south]{$2M$};
\foreach \y in {-1.5, 1.5}
\draw (2, \y) node[anchor = east]{$M$};
\draw (2.25, 0) node[anchor = north]{$x_0$};

\foreach \x/\lbl in {1/$\alpha$, 3/$\beta$}
\draw (\x, .1) -- (\x, -.1) node[anchor = north]{\lbl};

\foreach \x in {-1.12, -.34, -.08, -1.63, -.61, .85, -1.42,
4.41, 4.07, 3.81}
\fill (\x, 0) circle(.075);
\foreach \x in {-.75, -2.3, -1.71, -.44, .32, -1.89, .14,
3.52, 3.66, 4.23}
\draw (\x, 0) circle(.075);
\draw (-.5, .2) node[anchor = south]{$\lambda_n, \ldots, \lambda_{r + 1}$};
\draw (4, .2) node[anchor = south]{$\lambda_r, \ldots, \lambda_1$};

\end{tikzpicture}
\caption{The largest $r$ eigenvalues of both $X$ ($\circ$) and $Y$ ($\bullet$) are in the interior of $\gamma$, while the remaining $n - r$ eigenvalues are in the exterior.}\label{cauchypic}
\end{figure}

\end{document}